\documentclass[pra,reprint]{revtex4-2}
\usepackage[english]{babel}
\usepackage[utf8]{inputenc}
\usepackage{amsmath,amssymb,amsthm}
\usepackage{mathrsfs, dsfont}
\usepackage{mathtools}
\usepackage[normalem]{ulem}
\usepackage{braket}
\usepackage{enumitem}

\usepackage{tikz}
\usetikzlibrary{positioning,calc,fadings,decorations.pathreplacing,shadings}
\usetikzlibrary{arrows}
\usetikzlibrary{shapes,backgrounds,positioning,fit}
\usetikzlibrary{perspective}

\usepackage{hyperref}
\hypersetup{
    colorlinks=false, %set true if you want colored links
    linktoc=all,     %set to all if you want both Sections and subSections linked
    linkcolor=black,}

\newtheorem{thm}{Theorem}[section]

\newtheorem{Cor}{Corollary}

\newtheorem{proposition}[thm]{Proposition}

\newtheorem{remark}{Remark}

\newcommand{\cV}{\ensuremath{\mathcal V}}

 \let\b=\beta   
\let\d=\delta  

 \let\g=\gamma       \let\l=\lambda
            
\let\r=\rho  \let\s=\sigma    
  
   \let\G=\Gamma

\def\bR{{\mathbb R}}
\def\bC{{\mathbb C}}

\newcommand*{\Nbas}{N_b}
\newcommand*{\HermitanMat}{\mathbb{H}}
\newcommand*{\oneH}{\mathcal{H}}				%one-particle Hilbert space
\newcommand*{\NTdensMat}{\mathcal{P}_{N,\pm}}		%set of all finite T density-matrix operators, w > 0
\newcommand*{\NdensMat}{\overline{\mathcal{P}}_{N,\pm}}	%set of all density-matrix operators
\newcommand*{\ToneMat}{\mathscr{N}_{N,\pm}}
\newcommand*{\NoneMat}{\overline{\mathscr{N}}_{N,\pm}}
\newcommand*{\VoneMat}{\mathscr{V}_{N,\pm}}			

\newcommand*{\Helmholtz}{\Omega}
\newcommand*{\HelmholtzV}{\Omega_{N,\pm}}
\newcommand*{\uniF}{F_{N,\pm}}

\newcommand*{\crea}[1]{\hat{#1}^{\dagger}}
\newcommand*{\anni}[1]{\hat{#1}^{\vphantom{\dagger}}}

\DeclareMathOperator{\sgn}{sgn}
\DeclareMathOperator{\trace}{tr}
\DeclareMathOperator{\Trace}{Tr}

\DeclarePairedDelimiter{\abs}{\lvert}{\rvert}
\DeclarePairedDelimiter{\norm}{\lVert}{\rVert}

\newcommand*{\isDefinedAs}{\coloneqq}
\newcommand*{\ud}{\mathrm{d}}

\allowdisplaybreaks[4]

\begin{document}

\title{One-body reduced density-matrix functional theory for the canonical ensemble}
\date{\today}

\author{S.M. Sutter}
\email{s.m.sutter@vu.nl}
\author{K.J.H. Giesbertz}
\email{k.j.h.giesbertz@vu.nl}
\affiliation{Theoretical Chemistry, Faculty of Exact Sciences, VU University, De Boelelaan 1083, 1081 HV Amsterdam, The Netherlands}

\begin{abstract}
We establish one-body reduced density-matrix functional theory for the canonical ensemble in a finite basis set at an elevated temperature. Including temperature guarantees differentiability of the universal functional by occupying all states and additionally not fully occupying the states in a fermionic system. We use convexity of the universal functional and invertibility of the potential-to-1RDM map to show that the subgradient contains only one element which is equivalent to differentiability.
This allows us to show that all 1RDMs with a purely fractional occupation number spectrum ($0 < n_i < 1 \; \forall_i$) are uniquely $v$-representable up to a constant.
\end{abstract}

\maketitle

\section{Introduction}
Quantum chemistry and physics deal with the description of many interacting particles. Often we limit ourselves to a single particle species. In quantum chemistry these are usually electrons, but in physics also bosonic particles are of interest. Though the many-body Schrödinger equation involves only linear operators, the daunting dimensionality of the many-body wave function renders a direct solution intractable, but for a few particles. This is one of the prime reasons to aim directly for reduced quantities.

In 1964 Hohenberg and Kohn presented their revolutionary work about density functional theory (DFT)~\cite{hohenberg1964inhomogeneous}. They showed that any observable can be regarded as a functional of the density. Especially the Kohn--Sham (KS) formulation~\cite{kohn-sham1965} has been important to the success of DFT. Their idea was to approximate the true kinetic energy by the kinetic energy of the KS system: a non-interacting system with the same density as the interacting system.  
The KS kinetic energy turns out to be a decent approximation to the true kinetic energy of the interacting system. The difference in the kinetic energy is then lumped together with the interaction beyond Hartree (classical Coulomb) in the exchange-correlation energy functional. Though formally exact, in practice KS-DFT has some weaknesses, since the exchange-correlation functional needs to be approximated. A famous example is the stretching of the H\textsubscript{2} bond~\cite{cohen2012challenges, becke2014j, vuckovic2015hydrogen}.

One way to bypass some of these problems in constructing an approximate exchange-correlation energy functional is (one-body) reduced density matrix (1RDM) functional theory. One advantage over DFT is that we have also an explicit expression for the kinetic energy while still having the total energy as a functional of the 1RDM~\cite{gilbert1975hohenberg}. However, in the zero temperature setting, mapping back from 1RDMs to (non-local) potentials is problematic, as already noted by Gilbert~\cite{gilbert1975hohenberg} and others~\cite{pernal2005, Leeuwen2007, PhD-Baldsiefen2012, giesbertz2019one}. This is most clear in the case of non-interacting particles, since typically ground state 1RDMs are idempotent. It therefore seems that non-idempotent 1RDMs cannot be $v$-representable in the absence of interactions. There is the possibility for orbital energies to be degenerate, however, which allows fractionally occupied orbitals and hence non-idempotent 1RDMs~\cite{RequistPankratov2008, GiesbertzBaerends2010}. But one quickly realizes that the scaled identity operator is the one-body Hamiltonian which has all 1RDMs as ground state 1RDM. It is clear that the interaction should play a crucial role in the back mapping, but there has been no progress in this direction.

An alternative to regularize the theory is to introduce entropy, i.e., work at finite temperature, as proposed more than a decade ago~\cite{Leeuwen2007, PhD-Baldsiefen2012, baldsiefen2013minimization, baldsiefen2015reduced}
Though this is a theoretical motivation to introduce temperature, also physically this is a well justified choice, since most experiments are conducted at $T > 0$. Important examples where temperature plays an important role are metal-insulator transitions in transition metal oxides~\cite{YooMaddoxKlepeis2005}, high $T_c$ super conductors~\cite{NagamatsuNakagawaMuranaka2001}, hot plasmas~\cite{Dharma-wardanaPerrot1982}, etc.

In Ref.~\cite{giesbertz2019one} 1RDM functional theory (1RDMFT) was presented for the grand canonical ensemble within a finite basis set. However, the use of a grand canonical ensemble is inappropriate if the number of particles is relatively low as in ultra cold atom experiments~\cite{EwaldFeldkerHirzler2019}, but also in the low temperature limit the grand canonical ensemble can lead to unphysical results~\cite{Bedingham2003, MullinFernandez2003}. A canonical formulation of 1RDMFT is therefore desirable and will be the goal of this article.

In classical thermodynamics the grand potential can be reached by a Legendre transformation of the Helmholtz free energy with respect to the number of particles. In the quantum mechanical setting we can not do that. The reason is that the grand potential and the Helmholtz functional act on different spaces: the Fock space and the $N$-particle Hilbert spaces, respectively. Thus, we can not simply transform it back to obtain the canonical case.
Another major difference w.r.t.\ the grand canonical ensemble is that in the non-interacting case the occupation numbers are not explicitly given by either the Fermi/Bose function for fermions/bosons. Instead, they need to be calculated recursively, using auxiliary partition functions~\cite{barghathi2020theory}.
In Ref.~\cite{baldsiefen2015reduced} it therefore remained an open question whether every thermal 1RDM (only fractional occupation numbers) would be non-interacting $v$-representable and not much progress could be made.
However, we do not rely on such an explicit relation and we are able to prove a one-to-one correspondence between thermal 1RDMs (all occupation numbers fractional) and (non-local) potentials for any interaction. The non-interacting system is just a particular case.
This result justifies the existence of an algorithm which finds for any thermal 1RDM the corresponding non-interacting Hamiltonian as published recently by Kooi~\cite{kooi2022canonical}.

In this work we present 1RDMFT in a rigorous way for a fixed number of particles, finite basis set and elevated temperature. We show that the universal functional is differentiable and it holds
\begin{align}
\label{differentiability of universal functional}
    \frac{\partial \uniF}{\partial \g}=-v,
\end{align}
where $\g$ is the ground state 1RDM for the potential $v$. Here and in the following $+$ and $-$ stand for the bosonic and the fermionic case respectively. If we have a handy expression for $\uniF[\g]$ then we can circumvent the handling of the density-matrix operator to compute the free energy and instead we only need to deal with the reduced quantity $\g$. The minimizer for the Helmholtz functional can then be determined through the above relation \eqref{differentiability of universal functional}.

This work is built up in the following way. In Section~\ref{section setting} we introduce all the relevant spaces, then, in Section~\ref{section general approach}, we present the Helmholtz functional, its minimizing density-matrix operator and the general approach for 1RDMFT. For this task we make use of the universal functional $\uniF$. To show differentiability of $\uniF$ we utilize results from convex analysis. In Section~\ref{section general properties of the helmholtz functional and implications on the universal functional} we show that all the relevant functionals are convex. Additionally, we show that two potentials differing by more than a constant can not generate the same density-matrix operator. The proof of differentiability of $\uniF$ is finalized in Section~\ref{section final result}. 

\section{Setting}
\label{section setting}
We build our $N$-particle space from a finite number of single particle states $\ket{i}$, for $i \in \{1, \dotsc , \Nbas\}$ and $\Nbas < \infty$. We require the states to be orthonormal. The one-particle Hilbert space $\oneH$ is now the $\bC$-vector space generated by the states $\ket{i}$, i.e., $\oneH\isDefinedAs \mathrm{span}\{\ket{1}, \dotsc , \ket{\Nbas}\} \cong \bC^{\Nbas}$. To build the $N$-particle space we need to distinguish between bosons and fermions.
\begin{description}
\item[Bosons] A system with $N$ bosons is described by a symmetric wave function. Therefore, the bosonic $N$-particle Hilbert state, denoted by $\oneH^N_+$, consists of all symmetric tensors of order $N$, i.e.,  $\oneH_+^N\isDefinedAs\mathrm{Sym}^N(\oneH)$. The dimension of $\oneH_+^N$ is $\binom{\Nbas+N-1}{N}$.

\item[Fermions] Fermionic systems are described by anti-symmetric wave functions. Thus, we consider the space of anti-symmetric tensors of order $N$, i.e.\ $\oneH_-^N\isDefinedAs\wedge ^N \oneH$. The dimension is given by $\binom{\Nbas}{N}$. Note that we need to have $\Nbas \geq N$. The case $\Nbas=N$ is trivial since we have only one possible state. Thus, we will only consider $\Nbas>N$. 
\end{description}
The set of density-matrix operators on the $N$-particle space $\oneH_{\pm}^N$ is defined as 
\begin{equation}
    \label{definition set of density-matrix operators}
        \NdensMat \isDefinedAs \Set{\!\hat{\r}:\oneH_\pm^N\! \to \oneH_\pm^N | \hat{\r}=\hat{\r}^\dagger, \hat{\r}\geq 0, \Trace\{\hat{\r}\}=1\!\!},
\end{equation}
which we endow with the norm 
\begin{equation}
    \label{defintion norm for density-matrix operators}
        \norm{ \hat{\r} }_2= \big( \Trace\{\abs{\hat{\r}}^2\} \big)^{1/2}.
\end{equation}
A density-matrix operator $\hat{\r} \in \NdensMat$ has a spectral decomposition 
\begin{equation}
    \label{density operator in spectral decomposition}
        \hat{\r}= \sum \limits_l \l_l \ket{\psi_l}\bra{\psi_l},
\end{equation}
and its kernel is given by 
\begin{multline}
    \label{kernel of density-matrix operator}
        \r(x_1, \ldots, x_N; y_1, \ldots, y_N) \\
        = \sum \limits_l \l_l \psi_l(x_1, \ldots , x_N) \psi_l^*(y_1, \ldots, y_N).
\end{multline}
We can define the 1RDM $\g$ by \footnote{Creation and annihilation operator are only defined pairwise for a $N$ particle space.}
\begin{align*}
    \g_{ij}[\hat{\r}]=\Trace\{\hat{\r}\,\crea{a}_j \anni{a}_i\}. 
\end{align*}
It turns out (see Appendix) that the relevant spaces for the 1RDMs are subsets of the space of all Hermitian $\Nbas \times \Nbas$ matrices denoted by $\HermitanMat(\Nbas),$
\begin{subequations}
\begin{align}
    \label{definition spaces for 1RDM (bosons)}
        \overline{\mathscr{N}}_{N,+}\!&\isDefinedAs\! \Set{\! \g \in \HermitanMat(\Nbas) |\! \g \geq 0, \trace\{\g\}=N}, \\
        \label{definition spaces for 1RDM (fermions)}
        \overline{\mathscr{N}}_{N,-}\!&\isDefinedAs\! \Set{\! \g \in \HermitanMat(\Nbas) |\! \g \geq 0, \g^2\! \leq \g, \trace\{\g\}=N \!\!}.
\end{align}
\end{subequations}
We have used $\trace\{\cdot \}$ to emphasize that the trace is over the one-particle Hilbert space $\oneH$ as opposed to the $\Trace\{\cdot\}$ which is over a $\oneH^N_\pm$ Hilbert space.
By convention, the eigenvalues and eigenstates of the 1RDM $\g$ are called natural occupation numbers and natural orbitals (NO) respectively. Coleman has shown that all elements of $\NoneMat$ can be obtained from a density-matrix operator in $\NdensMat$, so it is a true 1RDM~\cite{coleman1963structure}.

\begin{thm}[Coleman]
\label{thm N representability of 1RDMs}
For any $\g \in \NoneMat$ there is a density matrix $\hat{\r} \in \NdensMat$ which generates $\g$.
\end{thm}
The proof can be found in the Appendix.

\section{General Approach}
\label{section general approach}
The Helmholtz functional for the canonical ensemble is defined as
\begin{equation}
    \label{definition Helmholtz functional}
        \Helmholtz_v[\hat{\rho}]\isDefinedAs E_v[\hat{\r}]-\b^{-1}S[\hat{\r}],
\end{equation}
where
\begin{equation}
    \label{energy functional}
        E_v[\hat{\r}]\isDefinedAs\Trace[\hat{\r}\hat H_v]
\end{equation}
is the energy of a system with Hamiltonian $\hat{H}_v\isDefinedAs\hat{H}_0+\hat{V}_v$ ($\hat H_0$ contains the kinetic and interaction part and $\hat V_v$ is the potential with kernel $v(x,x')$). The second term contains the entropy 
\begin{equation}
    \label{definition entropy}
        S[\hat{\r}]\isDefinedAs -\Trace\{\hat{\r}\log (\hat{\r})\}
\end{equation}
and the inverse temperature $\b=1/T$. With $\log$ we mean the natural logarithm.
The minimizer $\hat{\r}_v$ of the Helmholtz functional can be found by variations in the density-matrix operator which yields the equation
\begin{equation}
    \label{functional derivative of Helmholtz functional}
        \Trace\{\d \hat{\r}(\hat{H}_v+\b^{-1}\log(\hat{\r}_v))\}+\b^{-1}\Trace\big\{\d \hat{\r}\big\}=0.
\end{equation}
From the unit trace condition and \eqref{functional derivative of Helmholtz functional} it follows that
\begin{align}
    \label{minimizer of the Helmholtz functional}
        \hat{\r}_v &= e^{-\b\hat{H}_v}/Z[v], &
        &\text{where}& 
        Z[v] &\isDefinedAs \Trace\bigl\{e^{-\b \hat{H}_v}\bigr\}.
\end{align}
The minimizer $\hat{\r}_v$ is called Gibbs state.
Note that we only have a proper solution for $0<Z[v]<\infty$. This is always the case since we work in a finite basis setting with a fixed number of particles, so the trace only runs over a finite number of elements.

One aim is to show that the map from the potential $v$ to the density-matrix operator $\hat{\r}_v$ is invertible. However, this is only doable up to a constant since adding a constant to the potential does not change the density-matrix operator. To achieve a one-to-one correspondence we allow only potentials from the following set,
\begin{equation}
    \label{set of all potentials}
        \cV\isDefinedAs \Set{v \in \HermitanMat(\Nbas) | \trace\{v\}=0}.
\end{equation}
We can also think of $v\in \cV$ being a representative of the equivalence class containing potentials differing by a constant.

In Theorem~\ref{thm N representability of 1RDMs} we have seen that all $\g \in \NoneMat$ are $N$-representable. However, physically relevant are only the 1RDMs that are associated with a Gibbs state $\hat{\r}_v$. Thus, we denote the set of all $v$-representable 1RDM by
\begin{equation}
    \label{definition set of v-representable 1RDMs}
        \VoneMat \isDefinedAs \Set{ \g \in \NoneMat | \exists \; v \in \cV \mapsto \g}.
\end{equation}
The approach is to partition the minimization in the Helmholtz functional as 
\begin{equation}
    \label{infimum of the Helmholtz funcitonal over 1RDM }
        \HelmholtzV[v] \isDefinedAs \inf_{\mathclap{\hat{\r} \in \NdensMat}} \Helmholtz_v[\hat{\r}] 
        = \inf_{\mathclap{\g \in \NoneMat}} \big(\uniF[\g]+\trace\{v\g\}\big)
\end{equation}
where 
\begin{align}
    \label{defintion universal functional}
        \uniF[\g] \isDefinedAs{}& \inf_{\substack{\hat{\r}\in \NdensMat\\\hat{\r}\to \g}}\Helmholtz_0[\hat{\r}] \notag \\*
        {}={}&\inf_{\substack{\hat{\r}\in \NdensMat\\\hat{\r}\to \g}}\Trace\big\{\hat \r\big(\hat H_0 +\b^{-1}\log(\hat{\r})\big)\big\}
\end{align}
is called the universal functional which takes the value $\infty$ in case no $\hat{\r} \to \g$ exists. Here and in the following, $\Helmholtz_0[\hat{\r}]=\Helmholtz_{v=0}[\hat{\r}]$. The aim is to show that $\uniF$ is differentiable. Then the minimizer can be found through the relation 
\begin{equation*}
    \frac{\partial \uniF}{\partial \g}=-v
\end{equation*}
and we know that $\g$ is a canonical eq-1RDM (equilibrium 1-RDM) which was an open question in \cite{baldsiefen2015reduced}.
\section{General Properties of the Helmholtz Functional and Implications on the Universal Functional}
\label{section general properties of the helmholtz functional and implications on the universal functional}
\begin{thm}
\label{thm invertibility hamiltonian to density matrix}
The mapping $\hat{H}_v \mapsto \hat{\r}_v$ with $v\in \cV$ is invertible up to a constant in the Hamiltonian.
\end{thm}
\begin{proof}
Assume that two Hamiltonians $\hat{H}_v$ and $\hat{H}_v'$ differing in their potential yield the same density-matrix operator $\hat{\r}_v$. From~\eqref{functional derivative of Helmholtz functional} it follows that $\hat{\r}_v$ fulfills
\begin{align*}
    \frac{1}{\b}\log(\hat{\r}_v)+\hat{H}_v&=C,\\
    \frac{1}{\b}\log(\hat{\r}_v)+\hat{H}_v'&=C'.
\end{align*}
Subtracting these equations gives \(\hat{H}_v-\hat{H}_v' = C-C' \).
\end{proof}

\begin{remark}
Since we have a fixed number of particles in the Hilbert space, the constant in Theorem~\ref{thm invertibility hamiltonian to density matrix} can be of the form $ f(\hat{N})$ where $f: \bR \to \bR$, so this includes the arbitrary constant shift in the potential.
\end{remark}

\begin{Cor}
\label{cor v to density matrix is invertible}
The map $v \mapsto \hat{\r}_v$ with $v \in \cV$ is invertible.
\end{Cor}
Note that we only have a one-to-one correspondence because we require $\trace\{v\}=0$. Otherwise a constant shift in the potential would lead to the same density-matrix operator.

At this point we want to mention that the density-matrix operator $\hat{\r}_v$ is positive definite, $\hat{\r}_v >0$, and lies in the following subspace of $\NdensMat$,
\begin{equation}
    \label{space where statistical density operators lie}
        \NTdensMat\isDefinedAs \Set{ \hat{\r}: \oneH^N_\pm \to \oneH_\pm^N  | \hat{\r}=\hat{\r}^\dagger, \hat{\r}>0, \Trace\{\hat{\r}\}=1}.
\end{equation}
It follows that the natural occupation numbers $n_i$ are positive and in the fermionic case additionally $n_i<1$. To see this let $\phi_1, \dotsc , \phi_{\Nbas}$ be the NO basis and $\hat{\r}_v=\sum_j \l_j \ket{\psi_j}\bra{\psi_j}$ be the spectral decomposition of the density-matrix operator. Then, as the $\psi_j$'s build a basis of $\oneH_\pm^N$, each NO $\phi_i$ contributes to at least one of the eigenstates. So,
\begin{multline}
    \label{natural occupation number for statistical ensemble positiv}
        n_i= \sum \limits_j \l_j \int \ud x \ud y \ud x_2 \dotsi \ud x_N  \phi_i^*(x)\phi_i(y) \\
        \psi_j(x, x_2, \dotsc, x_N) \psi_j^*(y,x_2, \dotsc, x_N)>0,
\end{multline}
where we used the fact that all weights $\l_j=e^{-\b E_j}/Z$ are positive. In case of fermions we have already showed that $n_i \leq 1$. The $i$th NO can not be present in all $\psi_j$'s (in case $N \neq \Nbas)$, so
\begin{multline}
    \label{fermionic natural occupation number for statistical ensemble smaller than 1}
        n_i=\sum \limits_j \l_j \int \ud x \ud y \mathrm{x_2} \dotsi \ud x_N  \phi_i^*(x)\phi_i(y) \\
        \psi_j(x, x_2, \dotsc, x_N) \psi_j^*(y,x_2, \dotsc, x_N)<1, 
\end{multline}
because the integral is at least for one $j$ not equal to 1.
Summarized, the 1RDMs produced by a potential are contained in
\begin{align}
\mathscr{N}_{N,+}&\isDefinedAs \Set{ \g \in \HermitanMat(\Nbas) |  \g>0 }, \\
\mathscr{N}_{N,-}&\isDefinedAs \Set{ \g \in \HermitanMat(\Nbas) |  \g>0, \g^2<\g }.
\end{align}
We want to show that the most important functionals are either convex or concave to be able to use results from convex analysis. We start with the functional $\HelmholtzV[v]$ which is achieved through a minimization and thus turns out to be concave \cite{eschrig2010t}.
\begin{thm}
The functional $\HelmholtzV[v]= \min \limits_{\hat{\r} \in \NdensMat} \Helmholtz_v[\hat{\r}]=\min \limits_{\hat{\r} \in \NTdensMat} \Helmholtz_v[\hat{\r}]= -\b^{-1}\log\big(Z[v]\big)$ is strictly concave in $v$.
\end{thm}
\begin{widetext}
\begin{proof}
Let $v_1 \neq v_2$ be two potentials in $\cV$ and let $0<t<1$. Then we have
\begin{align*}
    \HelmholtzV[tv_1+(1-t)v_2]%&= \min_{\hat{\r}\in \NTdensMat}\Trace\Big\{\hat{\r}\big(t\hat{H}_{v_1}+(1-t)\hat{H}_{v_2}+\frac{1}{\b}\log(\hat{\r})\big)\Big\}\\
    &=\min_{\hat{\r}\in \NTdensMat}\Trace\Big\{\hat{\r}\big(t\hat{H}_{v_1} + t \frac{1}{\b}\log(\hat{\r})\big)\Big\}+\Trace\Big\{\hat{\r}\big((1-t)\hat{H}_{v_2}+(1-t)\frac{1}{\b}\log(\hat{\r})\big)\Big\}\\
    &>t \min_{\hat{\r}_1 \in \NTdensMat} \Trace\Big\{\hat{\r}_1\big(\hat{H}_{v_1}+\frac{1}{\b}\log(\hat{\r}_1)\big)\Big\} + 
    (1-t) \min_{\hat{\r}_2 \in \NTdensMat}\Trace\Big\{\hat{\r}_2\big(\hat{H}_{v_2}+\frac{1}{\b}\log(\hat{\r}_2)\big)\Big\}\\
    &=t \HelmholtzV[v_1]+(1-t)\HelmholtzV[v_2]
\end{align*}
where the strict inequality follows from Corollary~\ref{cor v to density matrix is invertible}.
\end{proof}
\end{widetext}
With Corollary~\ref{cor v to density matrix is invertible} it is possible to show a generalization of the Hohenberg--Kohn theorem for 1RDMs and non-local potentials~\cite{mermin1965thermal}.
\begin{thm}
\label{Mermin's theorem for canonical ensemble}
The map $v \mapsto \g_v$ for $v \in \cV$ is invertible.
\end{thm}
\begin{proof}
Assume there are two potentials $v_1\neq v_2 \in \cV$ yielding different density matrix operators $\hat{\r}_1 \neq \hat{\r}_2$ but the same 1RDM $\g$. Then we get
\begin{align*}
    \HelmholtzV[v_1]&=\Helmholtz_{v_1}[\hat{\r}_1]
    =\Helmholtz_{v_2}[\hat{\r}_1]+\trace\big\{\g(v_1-v_2)\big\} \\
    & > \Helmholtz_{v_2}[\hat{\r}_2]+\trace\big\{\g(v_1-v_2)\big\}\\
    &= \HelmholtzV[v_2]+\trace\big\{\g(v_1-v_2)\big\}.
\end{align*}
Changing the role of $v_1$ and $v_2$ and adding the two inequalities gives
\begin{align*}
    \HelmholtzV[v_1]+\HelmholtzV[v_2]<\HelmholtzV[v_2]+\HelmholtzV[v_1],
\end{align*}
which is a contradiction.
\end{proof}
One aim is to show that the universal functional is convex. For this purpose we first show that the entropy is strictly concave (\cite{ruelle1969statistical}, \cite{lieb1975convexity}, \cite{wehrl1978general}).
\begin{thm}
The entropy is strictly concave, i.e., for any $\hat{\r}_0, \hat{\r}_1 \in \NdensMat$ and $\l \in (0,1)$ we have $S[\l \hat{\r}_0+(1-\l)\hat{\r}_1]> \l S[\hat{\r}_0]+(1-\l)S[\hat{\r}_1]$.
\end{thm}
\begin{proof}
Let $\hat{\r}_\l= \l \hat{\r}_0+(1-\l)\hat{\r}_1 = \sum_k w_k \ket{\psi_k}\bra{\psi_k}$. We use strict concavity of the function $s(x)=-x\log(x)$ and we get
\begin{align*}
    S[\hat{\r}_\l]&=-\sum_k w_k \log(w_k) = \sum_k s\big(\braket{\psi_k|\hat{\r}_\l|\psi_k } \big) \\
    &= \sum_k s\big(\l \braket{ \psi_k|\hat{\r}_0|\psi_k } + 
    (1-\l) \braket{ \psi_k | \hat{\r}_1|\psi_k}\big)\\
    &> \l \sum_k s\big(\braket{ \psi_k| \hat{\r}_0|\psi_k } \big) +
    (1-\l) \sum_k s\big(\braket{ \psi_k|\hat{\r}_1|\psi_k} \big)\\ 
    &\geq \l \sum_k \braket{ \psi_k|s(\hat{\r}_0)|\psi_k } + 
    (1- \l) \sum_k \braket{ \psi_k |s(\hat{\r}_1)|\psi_k } \\
    &= \l S[\hat{\r}_0] + (1- \l) S[\hat{\r}_1]
\end{align*}
where we used Jensen's inequality for the last inequality.
\end{proof}
\begin{Cor}
The Helmholtz functional $\Helmholtz_v[\hat{\r}]$ is strictly convex in the density operator $\hat{\r}$. 
\end{Cor}
\begin{proof}
It follows directly from the fact that the Helmholtz functional is the sum of a linear and a strictly convex functional.
\end{proof}
\begin{thm}
\label{thm convexity of universal functional}
    The universal functional $\uniF[\g]$ is convex on $\NoneMat$.
\end{thm}
\begin{proof}
Let $\g_0, \g_1 \in \NoneMat$, $\l \in [0,1]$ and $\g_\l= \l \g_0+(1-\l)\g_1$ and taking $\hat{\r},\hat{\r}_0, \hat{\r}_1 \in \NdensMat$ we get
\begin{multline*}
  \l \uniF[\g_0] + (1-\l)\uniF[\g_1] \\*
  \begin{aligned}
    &= \l \inf_{\hat{\r}_0 \to \g_0} \Helmholtz_0[\hat{\r}_0] + 
    (1-\l) \inf_{\hat{\r}_1 \to \g_1} \Helmholtz_0[\hat{\r}_1] \\
    &= \inf_{\hat{\r}_0 \to \g_0} \inf_{\hat{\r}_1 \to \g_1}  \l \Helmholtz_0[\hat{\r}_0] + 
    (1-\l) \Helmholtz_0[\hat{\r}_1]\\
    &\geq \inf_{\hat{\r}_0 \to \g_1}\inf_{\hat{\r}_1 \to \g_1} \Helmholtz_0[\l \hat{\r}_0 +
    (1-\l)\hat{\r}_1] \\
    &= \inf_{\hat{\r}\to \g_\l} \Helmholtz_0[\hat{\r}] = \uniF[\g_\l]. &\qedhere
  \end{aligned}
\end{multline*}
\end{proof}

\section{Final Result}
\label{section final result}
Now we want to show that the universal functional $\uniF$ is differentiable. Differentiability is only defined on an open set. However, the set $\NoneMat$ has empty interior in $\HermitanMat(\Nbas)$. Thus, we need to embed $\NoneMat$ in a topological space where $\ToneMat$ is the interior of $\NoneMat$. The idea is to use the following result about sub\-gradients and sub\-differentials.

\begin{thm}
\label{properties of the subgradient/subdifferential}
Let $X$ be a finite dimensional vector space and let $f:X \to \bR \cup \{\infty\}$ be a convex function with domain $M$. Assume $M$ is contained in $a+\mathfrak{L}$ such that $\mathfrak{L}$ is a subspace with the lowest dimension such that there exists $a \in X$ with $M \subset a+ \mathfrak{L}$ \footnote{This means $M$ is contained in an affine space}. Let $\partial _\mathfrak{L}f(x)\isDefinedAs\partial f(x) \cap \mathfrak{L}$ where $\partial f(x)$ is the subdifferential of $f$ at a point $x$ in the set $M$. Then the following properties hold for $\partial_\mathfrak{L} f(x)$
\begin{enumerate}[label=(\roman*)]
\item the set $\partial_\mathfrak{L} f(x)$ is nonempty,
\item $f$ is differentiable at $x$ if and only if $\partial_\mathfrak{L} f(x)$ contains only one element. In that case this element equals the usual gradient. (With differentiable we mean that there is a linear map $J:\mathfrak{L}\to \bR$ such that for all $h \in \mathfrak{L}$ we have
$\lim_{h \to 0}\frac{1}{\norm{h}_\mathfrak{L}}\abs{f(x+h)-f(x)-J(h)}=0$. %, i.e.\ $\partial_ f(x)= \{ \nabla f(x)\}.$
%($ii$) & the set is compact and convex,  \\  
%($iii$)& for any $h \in X,$ $f_h'(x)=\max \{\langle d|h \rangle \; : \; d \in \partial f(x) \},$\\ 
\end{enumerate}
\end{thm}

\begin{figure}[t]
\begin{center}
\begin{tikzpicture}[3d view={60}{30}]
    \coordinate (O) at (0,0,0);
    \coordinate (A) at (3,0,0);
    \coordinate (B) at (0,3,0);
    \coordinate (C) at (0,0,3);
    \coordinate (M) at ($0.33*(A) + 0.33*(B) + 0.33*(C)$);
    
    \draw[->] (O) -- (4,0,0) node[pos=1.1]{$n_1$};
    \draw[->] (O) -- (0,4,0) node[pos=1.1]{$n_2$};
    \draw[->] (O) -- (0,0,4) node[pos=1.1]{$n_3$};
    
    \draw[->] (O) -- (M) node[midway, above left]{$a$};
    
    \filldraw[color=teal!20!white, opacity=0.5] (A) -- (B) -- (C) -- cycle;
    \filldraw [teal] (A) circle (1pt) node[below, color=black]{\small{$2$}};
    \filldraw [teal] (B) circle (1pt) node[above, color=black]{\small{$2$}};
    \filldraw [teal] (C) circle (1pt) node[left, color=black]{\small{$2$}};
%    \filldraw [teal] (M) circle (0.5pt);
    \draw[color=teal] (A) -- (B) -- (C) -- cycle;
    \draw[color=teal!40!white, opacity=0.5] (A) -- ($0.5*(B)+0.5*(C)$) -- cycle;
    \draw[color=teal!40!white, opacity=0.5] (B) -- ($0.5*(A)+0.5*(C)$) -- cycle;
    \draw[color=teal!40!white, opacity=0.5] (C) -- ($0.5*(B)+0.5*(A)$) -- cycle;
\end{tikzpicture}
%\begin{tikzpicture}[scale=1.7]
%    \draw[->] (0,0)--(0,2) node[pos=1.1]{$n_3$};
%    \draw[->] (0,0)--(1.85,0.6) node[pos=1.1]{$n_2$};
%    \draw[->] (0,0)--(1.2,-1.4) node[pos=1.1]{$n_1$};
%    \draw[->] (0,0)--(0.5,0.55) node[midway, above left]{$a$};
    
    %\draw[color=black!60!white] (0.35,-0.408335+0.7)--(0.35,-0.408335);
    %\draw[color=black!60!white] (1,0.324-0.521735)--(0.65,0.211)--(0.65,0.911);
    %\draw[color=black!60!white] (0.35,-0.408335+0.7)--(0,0.7)--(0.65,0.911);
    %\draw[color=black!60!white] (1,0.324-0.521735)--(0.35,-0.408335);
%    \filldraw[color=teal!20!white, opacity=0.5] (0,1.4)--(1.3,0.422)--(0.7,-0.81667)--(0,1.4);
%    \filldraw [teal] (0,1.4) circle (0.5pt) node[left, color=black]{\small{$2$}};
%    \filldraw [teal] (1.3,0.422) circle (0.5pt) node[above, color=black]{\small{$2$}};
%    \filldraw [teal] (0.7,-0.81667) circle (0.5pt) node[left, color=black]{\small{$2$}};
    %\filldraw[color=teal!60!white, opacity=0.5] (0.65,0.911)--(0.35,-0.408335+0.7)--(1,0.324-0.521735);
%    \draw[color=teal] (0,1.4)--(1.3,0.422)--(0.7,-0.81667)--(0,1.4);
    %\draw[color=black!60!white] (1,0.324-0.521735+0.7)--(1,0.324-0.521735);
    %\draw[color=black!60!white] (0.65,0.911)--(1,0.324-0.521735+0.7)--(0.35,-0.408335+0.7);
%\end{tikzpicture}
\end{center}
\caption{Representation of $\overline{\mathscr{N}}_{N,+}$ in terms of the occupation numbers for $\Nbas=3$ and $N=2$ (See proof of Theorem~\ref{thm N representability of 1RDMs}). The vector $a$ shifts a $2$ dimensional subspace to an affine space which contains $\overline{\mathscr{N}}_{N,+}$.}
\label{affine space of 1RDMs}
\end{figure}
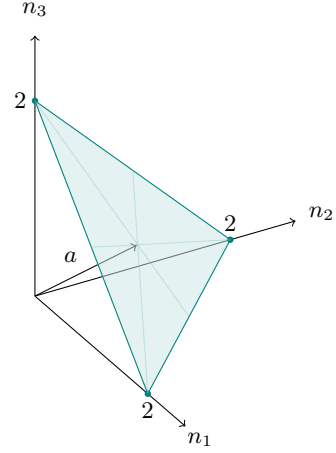

The universal functional has domain $\NoneMat$ which is contained in $\mathfrak{L}+a$  with $\mathfrak{L}=\Set{x \in \HermitanMat(\Nbas) | \trace\{x\}=0}$ and $a=N/\Nbas \cdot \mathds{1}$~\footnote{We could have taken a different point in $\mathscr{N}_{N, \pm}$, but the current choice is the most natural in this setting.} and its (relative) interior is $\mathscr{N}_{N, \pm}$. Fig. \ref{affine space of 1RDMs} shows $\overline{\mathscr{N}}_{N,+}$ for $2$ particles and $3$ basis functions. This also justifies the choice of the potential gauge. The aim is to get the relation $\partial \uniF/\partial \g=-v$. But as mentioned in Theorem \eqref{properties of the subgradient/subdifferential}, the differential is a map $J:\mathfrak{L}\to \bR$, i.e.\ it is contained in $\cV$. We can now apply the above theorem for all 1RDMs contained in $\ToneMat$. 
\begin{thm}
\label{main thm differentiation of F and v representability}
    If the infimum in \eqref{defintion universal functional} is attained, then
    \begin{enumerate}[label=(\roman*)]
    \item \( \ToneMat=\VoneMat \)\\
    \item the universal functional $\uniF[\g]$ is differentiable on $\ToneMat$.
    \end{enumerate}
\end{thm}
\begin{proof}
Convexity of $\uniF$ implies that for any $\g \in \ToneMat$ there exists at least one subgradient $h \in \cV$. So for all $\tilde{\g} \in \NoneMat$ it holds \(\uniF[\Tilde{\g}]+ \braket{ -h| \Tilde{\g} } \geq \uniF[\g]+\braket{ -h| \g } \), which implies 
\begin{align*}
    \uniF[\g] + \braket{ -h| \g } 
    \leq \min_{\mathclap{\tilde{\g}\in \NoneMat}} \uniF[\Tilde{\g}]+ \braket{ -h| \tilde{\g}}
    = \HelmholtzV[-h].
\end{align*}
Thus, the negative of the subgradient, $-h$, yields a potential generating $\g$ and hence \(\ToneMat=\VoneMat \). By Theorem~\ref{Mermin's theorem for canonical ensemble}, we get that there is only one such potential. Hence the subgradient is unique and $\uniF[\g]$ is differentiable for all $\g \in \ToneMat$ by Theorem~\ref{properties of the subgradient/subdifferential}. 
\end{proof}

We proved $v$-representability under the assumption that the minimum in~\eqref{defintion universal functional} is attained. To finish the proof we still need to show that this is indeed the case.
The idea is to show that the relevant functions are continuous and then use the fact that continuous functions attain their minima (and maxima) over compact sets.
\begin{proposition}
The energy $E_v[\hat{\r}]$ is Lipschitz continuous on $\NdensMat$.
\end{proposition}
\begin{proof}
The Hamiltonian acts on a finite dimensional space and thus it has a maximum eigenvalue, $\norm{\hat H_v}_\infty < \infty$. For two density-matrix operators $\hat{\r}_0, \hat{\r}_1$ we have
\begin{align*}
    \abs{E_v[\hat{\r}_0]-E_v[\hat{\r}_1]}
    &\leq \norm{\hat H_v}_\infty \norm{\hat{\r}_0- \hat{\r}_1}. \qedhere
\end{align*}
\end{proof}
\begin{proposition}
The entropy is continuous on $\NdensMat$.
\end{proposition}
\begin{proof}
The function $\hat{\r} \mapsto (a_1, \dotsc , a_m)$ where $(a_1, \dotsc, a_m)$ contains the eigenvalues of $\hat{\r}$ in an ordered fashion and with multiplicity ($a_i \leq a_{i+1}$ for all $i=1, \dotsc, m-1$) is continuous. The statement follows from continuity of the function $(a_1, \dotsc , a_m) \mapsto \sum_{j=1}^m a_j \log(a_j)$. 
\end{proof}
To finish the proof we make use of the following theorem.
\begin{thm}
Let $X$ be a compact metric space and let $f:X \to \bR$ be a continuous function. Then $f$ is bounded and it attains its maximum and minimum.
\end{thm}
Compact sets in finite dimensional affine spaces (with the usual metric) are fully characterized by closedness and boundedness. Thus, $\NdensMat$ and $\set{\hat{\r} \in \NdensMat | \hat{\r} \to \g }$ are compact and we get the following corollary.
\begin{Cor}
The infima in the the Helmholtz functional $\HelmholtzV[v]$ and the universal functional $\uniF[\g]$ are achieved in the fermionic and bosonic case, so the infima in \eqref{definition Helmholtz functional} and~\eqref{infimum of the Helmholtz funcitonal over 1RDM } can be replaced by minima.
\end{Cor}

\section{Conclusion}
For a fixed number of particles, finite numbers of basis functions and elevated temperature the universal functional in 1RDM functional theory is differentiable with $\partial \uniF/ \partial \g=-v$ for all 1RDMs $\g$ in $\ToneMat$. This relation holds for potential with $\trace\{v\}=0$. However,  all potentials differing from $v$ by only a constant lead to the same 1RDM. From this relation it follows directly that the map $v \mapsto \g$ is bijective up to a constant in the potential and it gives a characterization of the set of $v$-representable 1RDMs. Additionally, for every potential the Helmholtz functional and the universal functional attain a minimum. One might think of reaching the $T=0$ case by taking the limit $T \to 0$. The Gibbs state will just be an equi-ensemble of the ground states, but a difficulty is that it changes discontinuously when the potential is varied such that other states become ground states. Additionally, we can no longer guarantee that the 1RDM will be in the (relative) interior of $\NoneMat$ which prevents invertibility of $v \mapsto \g_v$ and also differentiability will probably no longer be in the cards.

\begin{acknowledgments}
The authors acknowledge support by the Netherlands Organisation for Scientific Research (NWO) under Vici grant 724.017.001.
\end{acknowledgments}

\appendix
\section{Properties of the 1RDM}
We are still left with showing that the set of 1RDMs is $\NoneMat$. To show the properties \eqref{definition spaces for 1RDM (bosons)} and \eqref{definition spaces for 1RDM (fermions)} for the 1RDM we use the following equivalent definition.

The kernel of $\g$ is given by tracing out $N-1$ particles in the density matrix operator,
\begin{equation}
    \label{definition 1RDM kernel}
        \g(x,y)= N \int \ud x_2 \dotsi \ud x_N \r(x, x_2, \dotsc , x_N;y, x_2, \dotsc, x_N).
\end{equation}
The 1RDM can be worked out in a $1$-particle orthonormal basis $\phi_1, \dotsc, \phi_{\Nbas}$ for a matrix representation with elements
\begin{multline}
    \label{1RDM in matrix representation}
        \g_{ij}=N \int \ud x \ud y \ud x_2 \dotsi \ud x_N \phi_i^*(x) \phi_j(y) \\* \r(x, x_2, \ldots, x_N; y, x_2, \ldots, x_N).
\end{multline}
Note that we have the following inequality,
\begin{multline}
    \label{}
        \g_{ii}= N \sum_l \l_l  \int \ud x_2 \dotsi \ud x_N \\
        \biggl( \int \ud x \phi_i^*(x) \psi_l(x, x_2, \dotsc, x_N) \biggr)^2 
        \geq 0.
\end{multline}
Next we want to show that for the fermionic 1RDM the diagonal entries are bounded from above by 1. For this, note that the entries of a statistical ensemble are bounded by the maximum value of the eigenstates of the corresponding density-matrix operator, i.e., 
\begin{align}
    \label{}
        \g_{ij} &=N \sum_{l} \l_l \int \ud x \ud y \ud x_2 \dotsi \ud x_N \phi_i^*(x) \phi_j(y) \notag \\*
        &\qquad \psi_l(x, x_2, \dotsc, x_N) \psi_l^*(y, x_2, \dotsc, x_N) \notag \\
        &\leq \max_l N \int \ud x \ud y \ud x_2 \dotsi \ud x_N \phi_i^*(x) \phi_j(y) \notag \\*
        &\qquad\psi_l(x, x_2, \dotsc, x_N) \psi_l^*(y, x_2, \dotsc, x_N) \notag \\
        &= \max \limits_l (\g_{\psi_l})_{ij} ,
\end{align}
where $\g_{\psi_l}$ means the 1RDM generated from the wave function $\psi_l$. Thus, we need to show the desired upper bound only for pure states. With a similar argument, it suffices to show the bound for Slater determinants built from $1$-particle orthonormal states $f_1, \dotsc, f_N$, %build from the $1$-particle basis $\phi_1, \ldots, \phi_{\Nbas}$
\begin{multline}
    \label{Slater determinant}
        \psi(x_1, \dotsc, x_N) \\
        = \frac{1}{\sqrt{N!}}\sum_{\s \in S_N} (-1)^{\sgn(\s)} f_{\s(1)}(x_1)\dotsb f_{\s(N)}(x_N).
\end{multline}
%where we have use $\phi_1, \ldots, \phi_N$ for simplicity. 
The kernel of the 1RDM can be worked out as
\begin{equation}
    \label{kernel of 1RDM of a Slater determinant}
        \g_\psi(x,y)=\sum \limits_{j=1}^N f_j(x) f_j^*(y)
\end{equation}
and its diagonal elements are
\begin{equation}
    \label{diagonal element of 1RDM of a Slater determinant}
        (\g_\psi)_{ii}=\sum \limits_{j=1}^N |\langle f_j| \phi_i \rangle|^2\leq \langle \phi_i| \phi_i \rangle =1.
\end{equation}
It is easy to see that the 1RDM is Hermitian. Thus, it has a spectral decomposition
\begin{equation}
    \label{spectral decomposition of 1RDM}
        \g=\sum_{l=1}^{\Nbas}\l_l \ket{\varphi_l}\bra{\varphi_l}.
\end{equation} 
Note that since the diagonal elements are non-negative for any basis, it follows that the eigenvalues $\l_l$ are non-negative. Therefore, $\g \geq 0$ and for fermions we have additionally $\g \leq 1$.
The trace of $\g$ can be calculated through its integral kernel \eqref{definition 1RDM kernel},
\begin{align}
    \label{trace of 1RDM}
        \trace\{\g\} &= \int \ud x \g(x,x) \notag \\
        &= N \sum_l \l_l \int \ud x\ud x_2 \dotsi \ud x_N\abs{\psi_l(x,x_2, \dotsc x_N)}^2 \notag \\
        &= N.
\end{align}
All these properties together show that the set of 1RDMs is contained in $\NoneMat$.

Next, we want to prove Theorem~\ref{thm N representability of 1RDMs}.
\begin{proof}
For $N=1$ we can simply take $\hat{\r}=\g$. So let us consider the case $N\geq2$. We represent $\g$ in the NO basis $\g = \sum_{j=1}^{\Nbas} \l_j \ket{\varphi_j}\bra{\varphi_j}$. We need to distinguish between bosons and fermions.
\begin{description}
\item[Bosonic case] We define the $N$-particle wave function 
\begin{equation*}
    \psi(x_1, \ldots, x_N)\isDefinedAs \frac{1}{\sqrt{N}}\sum_{j=1}^{\Nbas} \l_j^{1/2} \prod_{i=1}^N \varphi_j(x_i).
\end{equation*}
It is now easy to see that $\psi$ is symmetric, normalized and that it generates $\g$.

\item[Fermionic case] We work with a polytope. The 1RDM $\g$ can be expressed as a vector of length $\Nbas$ containing its occupation numbers $\textbf{n}=( \l_1, \dotsc, \l_{\Nbas})$. The extreme points of the polytope are all possible permutations of $N$ occupation numbers set to one and all other set to zero
\begin{equation*}
    \overline{\g}_I\isDefinedAs \overline{\g}_{i_1 \ldots i_N}\isDefinedAs \textbf{e}_{i_1}+ \dotsb + \textbf{e}_{i_N}, 
\end{equation*}
for $1 \leq i_1 < \ldots < i_N \leq \Nbas$ and where the $\textbf{e}_i$'s are unit vectors. The index $I$ is a renumeration of $i_1 \ldots i_N$ and can take $K= \binom{\Nbas}{N}$ values. The vector $\textbf{n}$ is an element of the polytope
\begin{equation*}
    \G \isDefinedAs\Set{ \sum_{I=1}^K \mu_I \overline{\g}_I | \mu_I \geq 0, \; \sum \limits_{I=1}^K \mu_I=1 }.
\end{equation*}
The extreme points $\overline{\g}_{i_1 \dotsc i_N}$ can now be identified with $\ket{\varphi_{i_1} \dotsc \varphi_{i_N}}\bra{\varphi_{i_1} \dotsc \varphi_{i_N}}$. Since the mapping $\hat{\r}\to \g$ is linear we find that $\g$ is generated from a linear combination of the Slater determinants $\ket{\varphi_{i_1} \dotsc \varphi_{i_N}}\bra{\varphi_{i_1} \dotsc \varphi_{i_N}}$. \qedhere
\end{description}
\end{proof}

\bibliography{bibliography}

\end{document}